\documentclass{amsart}

\newtheorem{theorem}{Theorem}[section]
\newtheorem{lemma}{Lemma}

\theoremstyle{definition}

\theoremstyle{remark}
\newtheorem*{remark}{Remark}

\begin{document}

\title{Consistency of Regularization for Scalar Fields}

\author{Doug Pickrell}
\email{pickrell@math.arizona.edu}

\begin{abstract} In two dimensional constructive quantum
field theory for scalar fields, it is necessary to regularize both
the action and the total (Gaussian) volume. In this paper we
consider the compatibility of these regularizations.
\end{abstract}

\maketitle

\section{Introduction}

Throughout this note, $m_0>0$, and $P:\mathbb R\to \mathbb R$ is a
polynomial which is bounded from below. If $\hat{\Sigma}$ is a
closed Riemannian surface, then $\mathcal D(\hat{\Sigma })$ and
$\mathcal D'(\hat{\Sigma })$ denote the spaces of smooth and
generalized functions (or distributions) on $\hat{\Sigma}$,
respectively, and a function $f$ is identified with the
distribution $fdA$, where $dA$ is the Riemannian area form. The
pairing of a function $f$ and a distribution $\phi$ is denoted by
$(f,\phi )$.

The  $P(\phi)_2$ quantum field theory (in finite volume),
corresponding to the parameters $(m_0,P)$, is essentially defined
by the Feynman-Kac measure on $\mathcal D'(\hat{\Sigma })$ given
by
\begin{equation}\label{FKmeasure}exp(-\int_{\hat{\Sigma}}:P((\delta_x,\phi
)):_{C_0}dA(x))det_{\zeta} (\Delta
+m_0^2)^{-1/2}d\phi_C,\end{equation} where $\Delta$ denotes the
(nonnegative) Laplacian on $\hat{\Sigma}$, $det_{\zeta}$ denotes
the zeta function determinant, $C=(\Delta +m_0^2)^{-1}$, $d\phi_C$
is the Gaussian measure having Fourier transform
$$\int_{\mathcal D'(\hat{\Sigma })}e^{-i(f,\phi )}d\phi_C=exp(-\frac 1
2(f,Cf)),\quad f\in \mathcal D(\hat{\Sigma }),$$
$\int_{\hat{\Sigma}}:P((\delta_x,\phi )):_{C_0}$ denotes point
splitting regularization with respect to
$$C_0(m_0;x,y)=-\frac1{2\pi}ln(m_0 d(x,y)),$$
and $d(x,y)$ is the distance between the points
$x,y\in\hat{\Sigma}$ (see \cite{GJ}, which we use as a general
reference, and section 4 of \cite{Pi2}).

It is necessary to regularize $\int P((\delta_x,\phi))dA(x)$,
because the measure $d\phi_C$ is not supported on ordinary
functions. The use of $C_0$, rather than $C$, and the inclusion of
the zeta function determinant, are essential to show that the
measures in (\ref{FKmeasure}) lead to a theory which is local in
the sense of Segal; see \cite{Pi2}. This leads to the consistency
question addressed in the following

\begin{theorem}\label{maintheorem}Given a closed Riemannian surface $\hat{\Sigma}$, if
$m^2=m_1^2+m_0^2$, then
$$exp(-\frac 12m_1^2\int_{\hat{\Sigma}}:(\delta_x,\phi )^2:_{C_0(m_
0,\hat{\Sigma })})det_{\zeta}(\Delta
+m^2_0)^{-1/2}d\phi_{C(m_0,\hat{ \Sigma })}$$
$$=exp(\frac 1{4\pi}m_1^2(ln(m_0/4)+\gamma)A)det_{\zeta}(\Delta +m^2)^{-1/2}d\phi_{C(m,\hat{\Sigma })},$$
where $A=\int dA$ and $\gamma$ is Euler's constant.
\end{theorem}

Although I do not understand its significance, this Theorem
singles out $4exp(-\gamma)$ as a special value for the bare mass.

\section{Proof of Theorem \ref{maintheorem}}

Throughout this section $C=C(m_0,\hat{\Sigma })$ and
$C_0=C_0(m_0,\hat{\Sigma })$. It is a fundamental fact that
$$C=C(m_0;x,y)=C_0(m_0;x,y)+C_f(m_0;x,y),$$ where $C_f$ is a
smooth function of $(x,y)\in\hat{\Sigma}\times\hat{ \Sigma}$. We
refer to $C_f$ as the finite part of $C$. Let $\Delta
f_k=\lambda_kf_k$, where the $ f_k$ are normalized eigenfunctions
corresponding to the eigenvalues $0=\lambda_0<\lambda_1..$. We
will also write $\int ( \cdot )$ for the integral over
$\hat{\Sigma}$ with respect to $dA$.

We first recall that
\begin{equation}\label{C}\int :(\delta_x,\phi )^2:_C=\sum
(\phi_k^2-E(\phi_k^2))=\sum (\phi_ k^2-\frac
1{m_0^2+\lambda_k}),\end{equation} where $E(\cdot )$ denotes
expectation with respect to $d\phi_C$. To verify (\ref{C}), let
$\delta_{t,x}=exp(-t\Delta )\delta_x$. Then by definition (see
section 6.3 of \cite{GJ}) the left hand side is the limit as
$t\downarrow 0$ of
$$\int_x\{(\phi ,\delta_{t,x})^2-(\delta_{t,x},C\delta_{t,x})\}$$
$$=\int_x\{\int_y\int_z\delta_{t,x}(y)\delta_{t,x}(z)\phi (y)\phi
(z)-\int_w\delta_{t,x}(w)C\delta_{t,x}(w)\}$$
$$=\int_x\{\int_y\int_z\sum e^{-t(\lambda_j+\lambda_k)}f_j(x)f_j(y
)f_k(x)f_k(z)\phi (y)\phi (z)$$
$$-\int_w\sum e^{-t(\lambda_j+\lambda_k)}(m_0^2+\lambda_k)^{-1}f_
j(x)f_j(w)f_k(x)f_k(w)\}$$
$$=\int_x\{\sum e^{-t(\lambda_j+\lambda_k)}f_j(x)c_jf_k(x)\phi_k$$
$$-\int_w\sum e^{-t(\lambda_j+\lambda_k)}(m_0^2+\lambda_k)^{-1}f_
j(x)f_j(w)f_k(x)f_k(w)\}$$
$$=\sum (e^{-2t\lambda_j}\phi_j^2-e^{-2t\lambda_j}(m_0^2+\lambda_
j)^{-1}).$$ When we take the limit as $t\to 0$, we obtain
(\ref{C}). Thus
\begin{equation}\label{C0}\int :(\delta_x,\phi
)^2:_{C_0}=\sum (\phi_k^2-\frac 1{m_0^2+\lambda_ k})+\int
C_f(m_0,x,x).\end{equation}

We now claim that $$exp(-\frac 12m_1^2\int :(\delta_x,\phi
)^2:_{C_0})d\phi_{C(m_0,\hat{ \Sigma })}$$
\begin{equation}\label{measure} =exp(-\frac 12m_1^2\int
C_f(m_0,x,x))det_2(1+m_1^2C)^{-1/2}d\phi_{ C(m,\hat{\Sigma
})},\end{equation} where $det_2$ denotes the Hilbert-Schmidt
regularized determinant. This follows from (\ref{C0}) and the
following calculation:
$$exp(-\frac {m_1^2}2\sum (\phi_k^2-\frac 1{m_0^2+\lambda_k}))d\phi_
C$$
$$=exp(-\frac {m_1^2}2\sum (\phi_k^2-\frac 1{m_0^2+\lambda_k}))\prod_{
k=0}^{\infty}(\frac {m_0^2+\lambda_k}{2\pi})^{1/2}e^{-\frac {m^2_
0+\lambda_k}2\phi_k^2}d\phi_k$$
$$=\prod_{k=0}^{\infty}e^{\frac {m_1^2}2\frac 1{m_0^2+\lambda_k}}\left
(\frac {m_0^2+\lambda_k}{m^2+\lambda_k}\right)^{1/2}\prod_{k=0}^{
\infty}(\frac {m^2+\lambda_k}{2\pi})^{1/2}e^{-\frac {m^2+\lambda_
k}2\phi_k^2}d\phi_k$$
$$=\prod_{k=0}^{\infty}((1+\frac {m_1^2}{m_0^2+\lambda_k})e^{-\frac {
m_1^2}{m_0^2+\lambda_k}})^{-1/2}d\phi_{C(m,\hat{\Sigma
})}=det_2(1+m_1^2C)^{-1/2}d\phi_{ C(m,\hat{\Sigma })}.$$

Theorem \ref{maintheorem} is therefore equivalent to the following
statement about multiplicative anomalies for zeta function
determinants.

\begin{theorem}\label{multiplicative}
$$det_{\zeta}(m_1^2+m_0^2+\Delta )=det_{\zeta}(m_0^2+\Delta )det_
2(1+m_1^2C(m_0,\hat{\Sigma }))exp(m_1^2\int (C_f(m_0
;x,x)-\gamma_0)),$$ where $\gamma_0=\frac 1{2\pi}(ln(\frac 14
m_0)+\gamma)$.
\end{theorem}

\begin{proof} Let $E=m_0^2+\Delta=C^{-1}$. Then the left hand side
of Theorem \ref{multiplicative} equals
$$det_{\zeta}(m_1^2+E)=det_{\zeta}(E(1+m_1^2C)).$$
As a pseudodifferential operator, $C$ has order $-2$. Because
$\hat \Sigma$ is two dimensional, in general, if $\tilde {C}$ has
order $-2$, then $tr(E^{-s}\tilde {C})$ is holomorphic for
$Re(s)>0$ and has a Laurent expansion in a neighborhood of $s=0$
of the form
$$tr(E^{-s}\tilde {C})=Res(\tilde {C})\frac 1s+c_0+O(s),$$
where $Res$ is the noncommutative residue and the constant term
$c_0$ is called the finite part of the trace. In section 3 of
\cite{Pi1} we wrote $c_0=\mathcal F\mathcal Ptr(E^{-s}\tilde {C})$
(this does possibly depend upon the principal symbol of $E$, so
$E$ is included in the notation). By Lemma 3.10 of \cite{Pi1} the
above determinant equals
$$det_{\zeta}(E)exp(\mathcal F\mathcal Ptr(E^{-s}log(1+m_1^2C))).$$
Thus to prove Theorem \ref{multiplicative}, we need to show that
$$exp((\mathcal F\mathcal Ptr(C^slog(1+m_1^2C))-m_1^2\int (C_f-\gamma_0)))=det_2(1+m_
1^2C)$$ or \begin{equation}\label{eqn1}\mathcal F\mathcal
Ptr(C^slog(1+m_1^2C))-m_1^2\int (C_f-\gamma_0)=tr(log(1+m_1^2C)
-m_1^2C).\end{equation} It suffices to show this for $m_1^2$
sufficiently small. The right hand side of (\ref{eqn1}) equals
\begin{equation}\label{eqn2}tr(-\frac 12(m_1^2C)^2+\frac 13(m_1^2C)^3+...),\end{equation}
i.e. we just delete the first term in the expansion of the
logarithm.

We now consider the left hand side of (\ref{eqn1}). For $Re(s)>0$,
$$tr(C^slog(1+m_1^2C))-m_1^2\int (C_f-\gamma_0)$$
$$=tr(C^sm_1^2C)-m_1^2\int (C_f-\gamma_0)+tr(C^s(-\frac 12(m_1^2C)^2+\frac 1
3(m_1^2C)^3+...))$$ The third term extends to an analytic function
in a neighborhood of $s=0$, and its value at $s=0$ agrees with the
trace in (\ref{eqn2}). Thus to prove (\ref{eqn1}) we need to show
that for small $s$
$$tr(C^sm_1^2C)-m_1^2\int (C_f-\gamma_0)=\frac {Res(m_1^2C)}s+h(s)$$
where $h(s)$ is holomorphic in a neighborhood of $s=0$ and
vanishes at $s=0$. This is implied by the following lemma, which
is probably well-known to experts.

\begin{lemma}\label{mainlemma}
$$tr(C^sC)=\frac {Res(C)}s+\int (C_f-\gamma_0)+O(s),$$
where $Res(C)=A/4\pi$ and $A=\int dA$.
\end{lemma}

The left hand side of Lemma \ref{mainlemma} equals
\begin{equation}\label{trace}\frac 1{\Gamma
(s+1)}\int_0^{\infty}t^str(e^{-tE})dt.\end{equation} We use the
asymptotic expansion
$$tr[e^{-t\Delta}]=\frac 1{4\pi t}\int dA+\frac 1{12\pi}
\int r_gdA+O(t),\quad as\quad t\downarrow 0,$$ where $r_g$ denotes
scalar curvature. This expansion implies
$$tr[e^{-tE}]=(e^{-tm_0^2})(\frac 1{4\pi t}A+\frac 1{12\pi}\int_{
\Sigma}r_gdA+O(t))$$
$$=\frac 1{4\pi t}A+\frac 1{12\pi}\int_{\Sigma}r_gdA-\frac {m_0^
2}{4\pi}A+O(t)$$ as $t\downarrow 0$. Then (\ref{trace}) equals
$$\frac 1{\Gamma (s+1)}\int_0^{\infty}t^s\frac A{4\pi t}dt+\frac
1{\Gamma (s+1)}\int_0^{\infty}t^s\{tr(e^{-tE})-\frac A{4\pi t}\}d
t$$
$$=\frac A{4\pi}\frac 1s+\frac 1{\Gamma (s+1)}\int_
0^{\infty}t^s\{tr(e^{-tE})-\frac A{4\pi t}\}dt$$ The second term
is holomorphic in a neighborhood of $s=0$. Thus
\begin{equation}\label{eqn5}tr(C^{s+1})=\frac A{4\pi}\frac
1s+\int_0^{\infty}\{tr(e^{-tE})-\frac A{4\pi
t}\}dt+O(s)\end{equation} This implies that $Res(C)=A/{4\pi}$.

\begin{remark} The residue can be calculated in a second way. If $\alpha$
denotes the canonical one-form on $T^*\hat \Sigma$, and $\omega=d
\alpha$, then because the principal symbol of $C$ is $\vert p
\vert ^{-2}$ (as a function on $T^*\hat \Sigma$),
\begin{equation}\label{eqn4}Res(C)=\frac 1{(2\pi )^2}\int_{S(T^{*}\Sigma )}\vert p
\vert ^{-2}\alpha d\alpha=\frac 1{(2\pi )^2}\int_{\vert
p\vert^2\le 1}\omega\wedge \omega .\end{equation} In local
coordinates, if $p=\sum p_j dq^j$, then $\vert p \vert
^2=g^{i,j}p_ip_j$, and the integral in (\ref{eqn4}) has the local
expression
$$\frac 1{(2\pi )^2}\int_q (\int_{\vert
p\vert^2\le 1}dp_1dp_2)dq^1dq^2=\frac 1{(2\pi )^2}\int_q
\pi(det(g_{i,j})^{1/2})dq^1dq^2=\frac 1{4\pi}\int dA.$$
\end{remark}

Given (\ref{eqn5}), to complete the proof of the Lemma, we need to
show
\begin{equation}\label{eqn6}\int_0^{\infty}\{tr(e^{-tE})-\frac
A{4\pi t} \}dt=\int (C_f(x,x)-\gamma_0).\end{equation} The left
hand side equals

$$\lim_{T\downarrow 0}\int_T^{1/T}\{-\frac d{dt}tr(e^{-tE}C)-\frac
A{4\pi t}\}dt$$

$$=\lim_{T\downarrow 0}(tr(e^{-TE}C-e^{-T^{-1}E}C)+\frac A{2 \pi}ln
(T))$$

$$=tr(C_f)+\lim_{T\downarrow 0}(tr(e^{-TE}C_0)+\frac A{2 \pi}ln(T))$$

\begin{equation}\label{eqn7}=tr(C_f)+\lim_{T\downarrow 0}(tr(e^{-TE}(-\frac 1{2
\pi}ln(d(x,y)))+\frac A{2 \pi}ln(T))-\frac
A{4\pi}ln(m_0^2)).\end{equation}

We now calculate

$$\lim_{T\downarrow 0}(tr(e^{-TE}ln(d(x,y))-AlnT)$$

\begin{equation}\label{eqn8}=\lim_{T\downarrow 0}(\int_y
\int_xe^{-TE}(y,x)ln(d(x,y))-Aln(T)).\end{equation}

For small $T$ the double integral is concentrated near the
diagonal. For fixed $y$, there is an asymptotic expansion

$$e^{-TE}(y,x)=\frac 1{4 \pi T}exp(-d(x,y)/4T)(1+O(T))$$
In exponential coordinates centered at $y$, $dA=j(v)d\lambda(v)$,
where $j(v)=1+O(r^2)$, $r=\vert v \vert$, and $d\lambda(v)$
denotes the Riemannian volume for $v\in T_y$. Thus (\ref{eqn8})
equals

$$\lim_{T\downarrow 0}(\int_y \int_x\frac 1{4\pi
T}e^{-d(x,y)^2/4T}(1+O(T))ln(d(x,y))-Aln(T))$$

$$=\lim_{T\downarrow 0}(\int_y
\int_x\frac 1{4\pi T}e^{-d(x,y)^2/4T}ln(d(x,y))-Aln(T))$$

$$=\lim_{T\downarrow 0}(\int_y (
\int_{v\in T_y}\frac 1{4\pi
T}e^{-r^2/4T}ln(r)j(v)d\lambda(v)-ln(T)))$$

$$=A\lim_{T\downarrow 0}(\int_0^{\infty}
e^{-r^2/4T}(ln(r^2/4T)+ln(4T)d(r^2/4T)-ln(T)),$$ because $\int
\frac 1{4\pi T}e^{-r^2/4T}ln(r)O(r^2)d\lambda(v)\to 0$ as $T\to 0$

$$=A(ln(4)+(\int_0^{\infty}
e^{-u}(ln(u)du)=A(ln(4)-\gamma).$$

When we plug this into (\ref{eqn7}), we obtain (\ref{eqn6}). This
completes the proof of Lemma \ref{mainlemma} and Theorem
\ref{maintheorem}.

\end{proof}

\section{Conformally Invariant Background}

As in the previous section $\hat{\Sigma}$ is a Riemannian surface.
As in Section 4 of \cite{Pi2}, $d\phi_{C(0,\hat{\Sigma})}$ denotes
the infinite conformally invariant measure on generalized
functions $\phi=\sum \phi_n f_n$ given by

$$d\phi_{C(0,\hat{\Sigma})}=d\lambda (
\phi_0)\prod_{k=1}^{\infty}(\frac {\lambda_k}{2\pi})^{1/2}e^{
-\frac {\lambda_k}2\phi_k^2}d\phi_k$$

\begin{lemma}\label{massless}
$$\lim_{m_0\downarrow 0}det_{\zeta}(m_0^2+\Delta )^{-1/2}d\phi_{C
(m_0,\hat{\Sigma })}=\frac 1{\sqrt {2\pi}}det_{\zeta}'(\Delta )^{
-1/2}d\phi_{C(0,\hat{\Sigma })}$$
\end{lemma}

\begin{proof}On the one hand
$$det_{\zeta}(m_0^2+\Delta )^{-1/2}=exp(\frac 12\frac d{ds}\vert_{
s=0}\sum_{n=0}^{\infty}(m_0^2+\lambda_n)^{-s})$$

$$=m_0^{-1}exp(\frac 12\frac d{ds}\vert_{s=0}\sum_{n=1}^{\infty}(
m_0^2+\lambda_n)^{-s})$$

On the other hand
$$d\phi_{C(m_0,\hat{\Sigma })}=\prod_{k=0}^{\infty}(\frac {m_0^2+
\lambda_k}{2\pi})^{1/2}e^{-\frac {m^2_0+\lambda_k}2\phi_k^2}d\phi_
k$$
$$\frac {m_0}{\sqrt {2\pi}}exp(-\frac {m_0^2}2\phi_0^2)d\lambda (
\phi_0)\prod_{k=1}^{\infty}(\frac
{m_0^2+\lambda_k}{2\pi})^{1/2}e^{ -\frac
{m^2_0+\lambda_k}2\phi_k^2}d\phi_k$$

When we form the product, there is a cancellation involving $m_0$,
$$det_{\zeta}(m_0^2+\Delta )^{-1/2}d\phi_{C(m_0,\hat{\Sigma })}=$$
$$\frac 1{\sqrt {2\pi}}exp(\frac 12\frac d{ds}\vert_{s=0}\sum_{n=
1}^{\infty}(m_0^2+\lambda_n)^{-s})exp(-\frac
{m_0^2}2\phi_0^2)d\lambda (\phi_0)$$
$$\prod_{k=1}^{\infty}(\frac {m_0^2+\lambda_k}{2\pi})^{1/2}e^{-\frac {
m^2_0+\lambda_k}2\phi_k^2}d\phi_k$$ When we take the limit as
$m_0\to 0$, we obtain the right hand side of Lemma \ref{massless},
completing the proof.
\end{proof}

Recall that the measures $d\phi_{C(m,\hat{\Sigma})}$ are mutually
absolutely continuous, for $0\le m < \infty$; see Lemma 3 of
\cite{Pi2}. Suppose that $\sigma>0$. Then

$$e^{-\frac 12\sigma\int_{\hat{\Sigma}}:(\delta_x,\phi )^2:_{C_0(
m,\hat{\Sigma })}}det_{\zeta}(m_0^2+\Delta
)^{-1/2}d\phi_{C(m_0,\hat{ \Sigma })}$$
$$=e^{-\frac 12\sigma\int_{\hat{\Sigma}}:(\delta_x,\phi )^2:_{C_0
(m,\hat{\Sigma })}}e^{\frac 12\sigma\int_{\hat{\Sigma}}:(\delta_x
,\phi )^2:_{C_0(m_0,\hat{\Sigma })}}$$
$$e^{-\frac 12\sigma\int_{\hat{\Sigma}}:(\delta_x,\phi )^2:_{C_0(
m_0,\hat{\Sigma })}}det_{\zeta}(m_0^2+\Delta )^{-1/2}d\phi_{C(m_0
,\hat{\Sigma })}$$

$$=(\frac m{m_0})^{\frac {\sigma A}{4\pi}}exp(\sigma \gamma_0 A)det_{
\zeta}(\sigma +m_0^2+\Delta )^{-1/2}d\phi_{C(\sqrt {\sigma +m_0^2}
,\hat{\Sigma })}$$

$$=(mexp(\gamma-ln(4)))^{\frac {\sigma A}{4\pi}}det_{ \zeta}(\sigma +m_0^2+\Delta )^{-1/2}d\phi_{C(\sqrt
{\sigma +m_0^2} ,\hat{\Sigma })}$$

By taking the limit as $m_0 \to 0$, using Lemma \ref{massless}, we
obtain

\begin{theorem}\label{masslessbackground}
 $$e^{-\frac 12\sigma\int_{\hat{\Sigma}}:(\delta_x,\phi )^2:_{C_0(
m,\hat{\Sigma })}}det_{\zeta}'(\Delta
)^{-1/2}d\phi_{C(0,\hat{\Sigma } )}$$
$$=(\frac m4 exp(\gamma))^{\frac {\sigma A}{4\pi}}det_{ \zeta}(\sigma +\Delta )^{-1/2}d\phi_{C(\sqrt
{\sigma } ,\hat{\Sigma })}$$
\end{theorem}

In Segal's approach to qft, the map $\Sigma \to Area(\Sigma)$
defines an additive homomorphism from the category of Riemannian
surfaces to $\mathbb R$. This is the significance of the term
involving area.


\begin{thebibliography}{99}

\bibitem{GJ}J. Glimm and A. Jaffe, Quantum Physics, a Functional
Integral Point of View, Springer-Verlag (1981).

\bibitem{Pi1} D. Pickrell, \emph{On the action of the group of
diffeomorphisms of a surface on sections of the determinant line
bundle}, Pac. J. Math., Vol. 193, No. 1 (2000) 177-199.

\bibitem{Pi2} --------, \emph{$P(\phi)_2$ quantum field theories and
Segal's axioms}, Commun. Math. Phys. 280 (2008) 403-425.

\end{thebibliography}
\end{document}